\documentclass[11pt,a4]{article}

\usepackage{enumerate}
\usepackage{latexsym} 
\usepackage{theorem}
\usepackage{graphics}
\usepackage{graphicx}
\usepackage{amsmath}
\usepackage{amssymb}
\usepackage[english]{babel} 
\usepackage{comment}

\newtheorem{defeng}{Definition}[section]
\newtheorem{theorem}[defeng]{Theorem}

\newtheorem{lemma}[defeng]{Lemma}

\newtheorem{conjecture}[defeng]{Conjecture}
\newtheorem{corollary}[defeng]{Corollary}
\newtheorem{property}[defeng]{Property}
{\theorembodyfont{\rmfamily} }
{\theorembodyfont{\rmfamily} }
{\theorembodyfont{\rmfamily} }
{\theoremstyle{break}\theorembodyfont{\rmfamily} }
{\theoremstyle{break}\theorembodyfont{\rmfamily} }

\newcounter{claim}

\newenvironment{proof}[1][]%
 {\noindent {\setcounter{claim}{0}\sc proof ---
   }{#1}{}}{\hfill$\Box$\vspace{2ex}} 

\newenvironment{claim}[1][]%
{\refstepcounter{claim}\vspace{1ex}\noindent{(\it\arabic{claim}){#1}{}}\it}{\vspace{1ex}}

\newenvironment{proofclaim}[1][]%
	{\noindent {}{#1}{}}{ This proves~(\arabic{claim}).\vspace{1ex}}

\newcommand{\sm}{\setminus} 

\usepackage{ifpdf}

\ifpdf
\fi

\title{Excluding 4-wheels}

\author{Pierre Aboulker\thanks{LIAFA, Universit\'e Paris 7 --
     Paris Diderot (France),\ email: pierreaboulker@liafa.jussieu.fr.
    Partially supported by \emph{Agence Nationale de la Recherche} under reference
    \textsc{anr 10 jcjc 0204 01}.}
}
    
\begin{document}

\maketitle

\begin{abstract}
A 4-wheel is a graph formed by a cycle $C$ and a vertex not in $C$ that has at least four neighbors in $C$. We prove that a graph $G$ that does not contain a 4-wheel as a subgraph  is 4-colorable and we describe some structural properties of such a graph.
\end{abstract}

\section{Introduction}
A $wheel$ is a graph formed by a cycle $C$, called the $rim$, and a vertex $u$ (not in $V(C)$), called the $center$, such that $u$ has at least three neighbors in $C$. A wheel with rim $C$ and center $u$ is denoted by $(u,C)$. If $(u,C)$ is a wheel,  $v \in V(C)$ and $uv$ is an edge, we say that $uv$ is a $spoke$ of the wheel.
For any $k \ge 3$, a \textit{$k$-wheel} is a wheel with at least $k$ spokes.

We say that a graph $G$ {\em contains} a graph $H$ if $H$ is isomorphic to a subgraph of $G$. If a graph $G$ does not contain a graph $H$ we say that $G$ is \textit{H-free}. If $\mathcal H$ is a class of graphs we say that a graph $G$ is $\mathcal H$-free if for any $H \in \mathcal H$, $G$ is $H$-free. We say that two non adjacent vertices $u$ and $v$ are $twins$ if $N(u)=N(v)$.

In \cite{thomassen} Thomassen and Toft proved the following result. (An alternative proof is given in \cite{wheelfree}). 

\begin{theorem} \label{thomassenToft}
If $G$ is a 3-wheel-free graph, then either it contains a pair of twins or it contains a vertex of degree at most 2.
\end{theorem}

From this Theorem they easily get the following corollary that settles a conjecture proposed by Toft in \cite{toft}.

\begin{corollary}
If $G$ is a 3-wheel-free graph, then $G$ is 3-colorable.
\end{corollary}


 In \cite{turner} Turner proved the following general result on $k$-wheel-free graphs for any $k \ge 4$.

\begin{theorem} \label{thmturner}
For any integer $k \ge 4$, if $G$ is a $k$-wheel-free graph, then $G$ contains a vertex of degree at most k.
\end{theorem}

Note that the result stated in \cite{turner} is slightly weaker
than Theorem~\ref{thmturner}, but the proof given by Turner in
\cite{turner}  proves exactly the version given here.

Theorem \ref{thmturner} implies easily the following result.

\begin{corollary}
For any integer $k \ge 4$, if $G$ is a $k$-wheel-free graph, then $G$ is (k+1)-colorable.
\end{corollary}

In this paper we extend Theorem \ref{thomassenToft} and strengthen Theorem \ref{thmturner} for $k=4$ by proving the following result.

\begin{theorem} \label{main}
If $G$ is a 4-wheel-free graph, then either it contains a pair of twins or it contains a vertex of degree at most 3.
\end{theorem}  

Which have the following easy corollary.

\begin{corollary} \label{maincolor}
If $G$ is a 4-wheel-free graph, then $G$ is 4-colorable.
\end{corollary}

\begin{proof}
  We proceed by induction on the number of vertices of a 4-wheel-free graph $G$.  If
  $|V(G)| = 1$, then it is 4-colorable.  Otherwise, by Theorem \ref{main},
  either $G$ contains a vertex $w$ of degree at most~3, or a pair $\{u, v\}$
  of twins.  In the first case, we color $G\sm\{w\}$ by the induction
  hypothesis, and give to $w$ one of the four colors not used in its
  neighborhood.  In the second case,  we color $G\sm\{u\}$ by the
  induction hypothesis, and give to $u$ the same color as $v$.
\end{proof}

We conjecture the following : 

\begin{conjecture} \label{conj}
If $G$ is a $k$-wheel-free graph, then it is k-colorable.
\end{conjecture}

Note that concerning the coloring, Corollary \ref{maincolor} and Conjecture \ref{conj} are tight since for any $k \ge 4$, $K_k$ is a $k$-wheel-free graph. Another $k$-wheel-free graph ($k \ge 4$) of chromatic number $k$ can be built as follows : take two disjoint copies $H_1$ and $H_2$ of $K_{k-2}$, add a vertex $x$ complete to $H_1$ and $H_2$ and two adjacent vertices $a$ and $b$, such that $a$ is complete to $H_1$ and $b$ is complete to $H_2$. It is easy to see that the graph obtained in this way is indeed $k$-wheel-free and of chromatic number $k$. 
\\

Regarding this subject, one can also find some extremal result on 3-wheel-free graphs in \cite{thomassenExtremal} and on 4-wheel-free graphs in \cite{horev}.

\section{Terminology and notations}
All graphs in this paper are simple and undirected.
Let $G$ be a graph.
For $S \subseteq V(G)$, $G[S]$ denotes the subgraph
of $G$ induced by $S$, and $G\setminus S=G[V(G)\setminus S]$.
The graph $G$ is called \textit{k-connected} (for $k\in \mathbb N$) if $|V(G)| >k$ and $G \sm S$ is connected for any set $S \subseteq V(G)$ with $|S|<k$. The greatest integer $k$ such that $G$ is k-connected is the $connectivity$ $\kappa(G)$ of $G$.

A {\em path} $P$ is a sequence of distinct vertices $p_1p_2\ldots p_{\ell}$,
$\ell \geq 1$, such that $p_ip_{i+1}$ is an edge for all $1\leq i < \ell$. If $P=p_1p_2\dots p_{\ell}$ is a path, the vertices $p_2, \dots, p_{\ell-1}$ are called the \textit{internal vertices} of $P$.
Saying that $P=P[a,b]$ is a path means that $a$ and $b$ are its extremities. If $\{c,d\} \subseteq V(P)$, then $P[c,d]$ denotes the unique subpath of $P$ with extremities $c$ and $d$. We define $P]c,d]=P[c,d] \sm c$ and we define in a similar way $P[c,d[$ and $P]c,d[$. 
A cycle $C$ is a sequence of vertices $p_1p_2\ldots p_kp_1$, $k \geq 3$,
such that $p_1\ldots p_k$ is a path and $p_1p_k$ is an edge.

Let $G$ be a graph. 
If $F \subseteq V(G)$ and $x \in V(G)$, we denote by $N(x)$ the set of neighbors of $x$, by $N(F)$ the set of vertices from $V(G) \setminus F$ adjacent to at least one vertex of $F$ and we define $N_F(x)=N(x) \cap F$ the neighborhood of $x$ in $F$. 
We denote
by $\overline{F}$ the set $V(G) \setminus (F \cup N(F))$.
We say that $F$ is a \emph{fragment} of $G$ if $|N(F)|=\kappa(G)$ and
$\overline{F} \ne \emptyset$ (note that if $F$ is a fragment of $G$,
then $\overline{F}$ is a fragment too).  An \textit{end} of $G$ is a fragment not
containing any other fragments as a proper subset. It is clear that any
fragment $F$ contains an end, and that consequently all graphs contain
at least two disjoint ends : one in $F$, another one in
$\overline{F}$. Let $F$ be an end of $G$. If $|F|=1$ we say that $F$ is $trivial$. The graph $H$ induced by $F \cup N(F)$ plus all possible edges with both ends in $N(F)$ is called an \textit{end block} of $G$. The edges with both ends in $N(F)$ are called the \textit{marker edges} of $H$.

\section{Preliminaries}

Let $G$ be a graph, $k \ge 1$ an integer, $Y\subseteq V(G)$ a set of
at least $k$ vertices, and $x\in V(G) \sm Y$.  A family $F_x$ of
$k$ paths from $x$ to $Y$ whose only common vertex is $x$ and whose
internal vertices are not in $Y$, is called a \emph{k-fan from $x$
  to Y}. We define $\text{end}(F_x)=F_x \cap Y$.  The next result is classical (see  \cite{bondy.murty:book}).

\begin{lemma}[Fan Lemma]
  \label{fanLemma}
If $G$ is a k-connected graph, $x \in V(G)$ and $Y$ is a
  subset of $V(G) \sm \{x\}$ of cardinality at least $k$,
then there is a k-fan from $x$ to $Y$.
\end{lemma}

We will also need the following Lemma due to Perfect (see \cite{perfect}).

\begin{lemma} \label{perfect}
Let $G$ be a k-connected graph, $k \ge 2$, $k_1 \le k$, $x \in V(G)$, $S \subseteq V(G)$ with $|S|\ge k$ and $F^1_x$ a $k_1$-fan from $x$ to $S$. There exists a k-fan $F^2_x$ such that $\text{end}(F^1_x) \subseteq \text{end}(F^2_x)$.
\end{lemma}

The following theorem is a classical result due to Dirac (see \cite {dirac} or \cite{bondy.murty:book}).

\begin{theorem} \label{dirac}
In a k-connected graph $G$,
\begin{enumerate}
\item given any $k$ vertices, there is a cycle passing through the $k$ vertices;
\item given any edge and any $k-1$ vertices there is a cycle passing through all of them.
\end{enumerate}
\end{theorem}

In \cite{watkins}, Watkins and Mesner describe the structure around a set of $k$ vertices in a $(k-1)$-connected graph ($k \ge 4$) that are not contained in a cycle. This theorem is of  great use for the study of $k$-wheel-free graphs.

\begin{theorem} [see \cite{watkins}] \label{watkins} 
Let $G$ be a graph with $\kappa(G)=k-1 \ge 3$. If $X=\{x_1, \dots, x_k\} \subseteq V(G)$ has the property that no cycle of $G$ goes through all the vertices of $X$ then $G$ admits a cutset $S=\{s_1, \dots, s_{k-1}\}$ such that $S \cap X=\emptyset$ and $G\sm S$ has $k$ connected components $C_1, \dots, C_k$ such that $x_i \in C_i$ for $i= 1, \dots, k$.
\end{theorem}

Note that in \cite{watkins} Watkins and Mesner do not present their result in the same fashion as we do here, but they  prove exactly the version given here.


Let us finish this section by an easy and essential property on the end blocks of a graph.

\begin{property} \label{endBlock}
If $G$ is a graph of connectivity k and $F$ is a non trivial end of $G$, then the end block $G_F$ containing $F$ is (k+1)-connected.
\end{property}

\begin{proof}
Let $G$ be a graph of connectivity $k$, $F$ a non trivial end of $G$ and $G_F$ the end block of $G$ containing $F$. By way of contradiction suppose $E$ is a fragment of $H$ such that $|N(E)|=k$. Since $N(F)$ induces a clique, we may assume w.l.o.g.\ that $N(F) \subseteq E \cup N(E)$. Thus $\overline E$ is a fragment of $G$ properly included in $F$, a contradiction. 
\end{proof}

\section{Main results} \label{mainResults}

This section is subdivided into four subsections. The first three  subsections are dealing  with 4-wheel-free graphs of connectivity 4, 3 and 2 respectively, and the last subsection combined the obtained results to prove the announced statement on 4-wheel-free graphs.

We denote by $W(G)$ the set of vertices that are center of a 4-wheel in $G$.

\subsection{4-connected 4-wheel-free graphs} \label{subsection4}

We say that a graph $G$ is \textit{almost 4-wheel-free} if $|W(G)|\le 3$ and $W(G)$ induces a clique.
Almost 4-wheel-free graph are, in general, a super class of 4-wheel-free graphs, but we needed the result (see Theorem \ref{4connexe})  on almost 4-wheel-free graph in order to study the case of 4-wheel-free graphs of connectivity 3.

We first need two easy lemmas.

\begin{lemma} \label{5connexe}
If $G$ is a 5-connected graph, then $W(G)=V(G)$
\end{lemma}

\begin{proof}
Let $x$ be a vertex of $G$ and let $\{x_1, \dots, x_4\}$ be four neighbors of~ $x$. By Theorem \ref{dirac}, there is a cycle $C$ in $G \sm \{x\}$ going through $\{x_1, \dots, x_4\}$ and thus $(x,C)$ is a 4-wheel of $G$.
\end{proof}

\begin{lemma} \label{triangle}
If $G$ is a 4-connected graph, then every vertex contained in a triangle is in $W(G)$.
\end{lemma}

\begin{proof}
Let $abc$ be a triangle of $G$ and let $a_1, a_2$ be neighbors of $a$ different from $b$ and $c$. By Theorem \ref{dirac}, there is a cycle $C$ in $G \sm \{a\}$ passing though $\{a_1, a_2\}$ and $bc$.
\end{proof}

Let us now explain how  Theorem \ref{watkins} is to be applied to 4-connected  graphs with no 4-wheels. Let $G$ be a 4-connected graph,  let $x \in V(G) \sm W(G)$ and let $X=\{x_1,x_2,x_3,x_4\} \subseteq N(x)$. Since $x \notin W(G)$ there is no cycle going through all vertices of $X$ in $G \sm \{x\}$. So, by Theorem \ref{watkins} applied to $G \sm \{x\}$ there exists a set $S=\{s_1,s_2,s_3\}$ such that $x_i's$ for $i=1,2,3,4$ are in distinct connected components of $G \sm \{x,s_1,s_2,s_3\}$. We call the set $\{x,s_1,s_2,s_3\}$ a \textit{Watkins-Mesner-certificate} (\textit{WM-certificate} for short) for  $(x,\{x_1,x_2,x_3,x_4\})$ and we call $C(x_1)$, $C(x_2)$, $C(x_3)$, $C(x_4)$ the connected components containing respectively $x_1$, $x_2$, $x_3$ and $x_4$.

\begin{theorem} \label{4connexe}
If $G$ is a 4-connected almost 4-wheel-free graph then $G$ is isomorphic to $K_{4,4}$.
\end{theorem}

\begin{proof}
Let us argue by way of contradiction and suppose that $G$ is a 4-connected 4-wheel-free graph that is not $K_{4,4}$. By Lemma \ref{5connexe}, $\kappa(G)=4$.

\begin{claim} \label{K44}
$G$ does not contain $K_{4,4}$ as a subgraph.
\end{claim}

\begin{proofclaim}
Suppose by way of contradiction that $G$ contains a subgraph  $H=K_{4,4}$ with partition $\{x_1, x_2,x_3, x_4\}$ and $\{y_1, y_2,y_3, y_4\}$. Since $G \neq H$, we may assume that say $x_1$ has a neighbor $u \notin V(H)$. There is a 3-fan from $u$ to $H \sm x_1$ in $G \sm \{x_1\}$. If a path $P$ connects $u$ with a vertex from $\{x_2,x_3,x_4\}$ in $G \sm \{x_1\}$, then each vertex of $\{y_1, y_2,y_3, y_4\}$ is the center of a 4-wheel in $G$, a contradiction. So the three paths of the 3-fan end in $\{y_1, y_2,y_3, y_4\}$ which implies that each vertex of $\{x_1, x_2,x_3, x_4\}$ are centers of a 4-wheel in $G$, a contradiction.
\end{proofclaim}


\begin{claim} \label{degenerate}
Let $x \in V(G) \sm W(G)$ and let $\{x_1,x_2,x_3,x_4\} \subseteq N(x)$.
If $\{x,s_1,s_2,s_3\}$ is a WM-certificate for $(x,\{x_1,x_2,x_3,x_4\})$ and $C(x_i) = \{x_i\}$ for an $i \in \{1,2,3,4\}$, then $x_i \in W(G)$. 
\end{claim}

\begin{proofclaim}
Suppose w.l.o.g.\ that $C(x_1) = \{x_1\}$ and that $x_1 \notin W(G)$. So $N(x_1)=\{x,s_1,s_2,s_3\}$ and there is a WM-certificate $\{x_1,t_1,t_2,t_3\}$ of $(x_1, \{x,s_1,s_2,s_3\})$.
Now, for $i=2,3,4$, there is a 4-fan from $x_i$ to $\{x,s_1,s_2,s_3\}$ included in $C(x_i) \cup \{x,s_1,s_2,s_3\}$ which implies that $\{x_2,x_3,x_4\}= \{t_1,t_2,t_3\}$.  If $\{x_2,x_3,x_4\} \subseteq N(s_i)$ for $i=1,2,3$, then $G[\{x,s_1,s_2,s_3,x_1,x_2,x_3,x_4\}]$ contains a $K_{4,4}$, a contradiction to   (\ref{K44}). So we may assume that say $s_1$ have a neighbor $u \neq x_2$ in $C(x_2)$. There is a 3-fan from $u$ to $\{x,s_2,s_3\}$ in $G \sm \{s_1\}$ that is included in $C(x_2) \cup \{x,s_2,s_3\}$, so one of the paths of this fan links $s_1$ with a node in $\{x,s_2,s_3\}$ and avoids $\{x_1,x_2,x_3,x_4\}$, a contradiction.  
\end{proofclaim}

Let $x \in V(G) \sm W(G)$, let $X=\{x_1,x_2,x_3,x_4\} \subseteq N(x)$ and $S=\{x,s_1,s_2,s_3\}$ a WM-certificate for $(x,\{x_1,x_2,x_3,x_4\})$ such that $x$, $X$ and $S$ are chosen subject to the maximality of $|C(x_i)|$ where $x_i \notin W(G)$. Assume w.l.o.g.\ that $C(x_1)$ is the one that realizes the maximality.

By Lemma \ref{triangle} at most one neighbor of $x$ is in $W(G)$ so we may assume w.l.o.g.\ that $x_2 \notin W(G)$ and thus, by \ref{degenerate}, $C(x_2) \neq \{x_2\}$.
Let $\{y_1,y_2,y_3\} \subseteq N(x_2) \sm \{x\}$ and let $\{x_2,t_1,t_2,t_3\}$ be a WM-certificate for $(x_2, \{x,y_1,y_2,y_3\})$

Suppose that for some $j \in \{1,2,3\}$, $C(y_j)$ is not included in $C(x_2)$ and therefore contains an $s_i$ ($i \in \{1,2,3\}$). W.l.o.g.\ $s_1 \in C(y_1)$. Since there are three internally disjoint paths linking $s_1$ to $x$ whose interior vertices are included  respectively in $C(x_1)$, $C(x_3)$ and $C(x_4)$, $\{t_1,t_2,t_3\} \subseteq (C(x_1) \cup C(x_3) \cup C(x_4))$. Now, if for an $i \in \{2,3\}$ $C(y_i)$ does not contain $s_2$ nor $s_3$ then $x_2$ is a cutvertex of $G$ separating $C(y_i)$ from the rest of the graph. 
So we may assume w.l.o.g.\ that $s_2 \in C(y_2)$ and $s_3 \in C(y_3)$. Since $|C(x_2)| \ge 2$, either $C(x)$ or one of the $C(y_i)$'s have at least one vertex in $C(x_2)$ and thus, either $\{x,x_2\}$ or $\{s_i,x_2\}$ is a cutset of $G$, a contradiction. 
So, $C(y_i) \subseteq C(x_2)$ for $i=1,2,3$ and therefore $C(x_1) \cup C(x_3) \cup C(x_4) \subseteq C(x)$. This contradicts the maximality of $C(x_1)$.
\end{proof}

Note that, after adapting definitions, the exact same proof works to show that for any $k \ge 5$, the only k-connected $k$-wheel-free graph is $K_{k,k}$.

\subsection{4-wheel-free graphs of connectivity 3} \label{subsection3}

\begin{theorem} \label{3connexe}
If $G$ is a graph with $\kappa (G)=3$ and $F$ is an end of $G$ such that $W(G) \cap F = \emptyset$, then $F$ is trivial.
\end{theorem}

\begin{proof}
Let $G$ be a graph of connectivity 3, let $F$ be an end of $G$ such that $W(G) \cap F = \emptyset$ and $N(F)=\{a_1,a_2,a_3\}$ and let $H$ be the end block containing $F$. Remind that edges $a_1a_2$, $a_2a_3$ and $a_1a_3$ are called the marker edges of $H$. If $F$ is trivial we are done, so assume that $|F| \ge 2$. By Property~ \ref{endBlock}, $H$ is 4-connected, and since $H$ contains a triangle, $H \neq K_{4,4}$. So by Theorem \ref{4connexe} $H$ admits a 4-wheel $(x,C)$ such that $x \in F$. It suffices now to show that $x$ is also the center of a 4-wheel in $G$ to get a contradiction.

The cycle $C$ has to contain some marker edges.
If $C$ contains exactly  one marker edge we replace this edge by a path of $\overline F \cup N(F)$ to get a 4-wheel centered on $x$ in $G$ (observe that it suffices to find a 4-wheel centered on $x$ in $H$ whose contains at most one marker edge). 
If $x \notin N_F(\{a_1,a_2,a_3\})$ and $C$ contains two marker edges then they are consecutive on $C$ and thus we may replace them by the third marker edge. 

So we may assume that $x$ is a neighbor of say $a_2$ and that $\{a_1a_2,a_2a_3\} \subseteq E(C)$. Put $P=C \sm \{a_2\}$. Let $x_1$, $x_2$, $x_3$ be the three neighbors of $x$ in $P$ and suppose that $a_1$, $x_1$, $x_2$, $x_3$, $a_3$ appear in this order along $P$. 

There is a 1-fan $Q_{a_2-v}$ from $\{a_2\}$ to $P\sm \{a_1,a_3\}$ in $H \sm\{a_1,a_3,x\}$. Let $v$ be the extremity of $Q_{a_2-v}$ different from $a_2$. 

If $v \in V(P[a_1,x_1])$ then $Q_{a_2-v} \cup P[v,a_3] \cup a_3a_2$ is the rim of a 4-wheel centered on $x$ in $H$ that uses exactly one marker edge, a contradiction. So $v \notin V(P[a_1,x_1])$ and symmetrically $v \notin V(P[x_3,x_4])$. Therefore we may assume w.l.o.g.\ that $v \in P]x_1,x_2]$.

Observe that $\{P[x_1,a_1], P[x_1,v]\}$ is a 2-fan  in $H \sm \{x\}$ from $x_1$ to $Q_{a_2-v} \cup P[v,a_3] \cup \{a_1\}$. So, by Lemma \ref{perfect}, there is a 3-fan $F_{x_1}=\{Q_{x_1-a_1},Q_{x_1-v},Q_{x_1-w}\} $ in $H \sm \{x\}$ from $x_1$ to $Q_{a_2-v} \cup P[v,a_3] \cup \{a_1\}$ with $\{a_1,v\} \subseteq \text{end}(F_{x_1})$. Assume that the extremity of $Q_{x_1-a_1}$ (resp.\ $Q_{x_1-v}$, resp.\ $Q_{x_1-w}$) that is not $x_1$ is $a_1$ (resp.\ $v$, resp.\ $w$). 
We alter $P$ as follow : $P[x_1,a_1] :=Q_{x_1-a_1}$ and $P[x_1,v]:=Q_{x_1-v}$. 
Note that, after this alteration, $P$ is still a path and $P \cup \{a_2\}$ is still the rim of a 4-wheel centered on $x$ in $H$ with spokes $xa_2$, $xx_1$, $xx_2$ and $xx_3$.

Let us now prove that $w \in P]x_2,x_3[$.
\begin{itemize}
\item $w \notin Q_{a_2-v}$ for otherwise $Q_{a_2-v}[a_2,w] \cup Q_{x_1-w} \cup P[x_1,v] \cup P[v,a_3] \cup a_3a_2$ is the rim of a 4-wheel centered on $x$ in $H$ that uses exactly on marker edge.
\item $w \notin P]v,x_2]$ for otherwise $Q_{x_1-w} \cup P[w,a_3] \cup a_3a_2 \cup Q_{a_2-v} \cup P[x_1,v]$ is the rim of a 4-wheel centered on $x$ in $H$ that uses exactly on marker edge.
\item $w \notin P[x_3,a_3]$ for otherwise $Q_{x_1-w} \cup P[w,v] \cup Q_{a_2-v} \cup a_2a_1 \cup P[x_1,a_1]$ is the rim of a 4-wheel centered on $x$ in $H$ that uses exactly on marker edge.
\end{itemize}

So $w \in P]x_2,x_3[$. Observe that now if $v = x_2$, then $x$ is the center of a 4-wheel in $H$ such that the rim uses exactly one marker edge, so $v \neq x_2$. 

Now, $\{P[x_2,v], P[x_2,w]\}$ is a 2-fan from $x_2$ to 
$Q_{a_2-v} \cup Q_{x_1-w} \cup P[a_1,v] \cup P[w,a_3]$ in $H \sm \{x\}$. 
So, by Lemma \ref{perfect}, there is a 3-fan $F_{x_2} = \{Q_{x_1-v},Q_{x_2-w},Q_{x_2-u}\}$ 
from $x_2$ to $Q_{a_2-v} \cup Q_{x_1-w} \cup P[a_1,v] \cup P[w,a_3]$ 
in $H \sm \{x\}$ such that $\{v,w\} \subseteq \text{end}(F_{x_2})$. 
Suppose that the extremity of $Q_{x_1-v}$ (resp.\ $Q_{x_2-w}$, resp.\ $Q_{x_2-u}$) 
that is not $x_1$ is $v$ (resp.\ $w$, resp.\ $u$).
We alter $P$ as follow : $P[x_2,v] :=Q_{x_2-v}$ and $P[x_2,w]:=Q_{x_2-w}$. 
Note that, after this alteration, $P$ is still a path and $P \cup \{a_2\}$ is still the rim of a 4-wheel centered on $x$ in $H$ with spokes $xa_2$, $xx_1$, $xx_2$ and $xx_3$.
Let us show by discussing among the position of $u$ in 
$Q_{a_2-v} \cup Q_{x_1-w} \cup P[a_1,v] \cup P[w,a_3] $ 
that $x$ is the center of a 4-wheel of $H$ such that the rim uses at most one marker edge.
\begin{itemize}
\item If $ u \in Q_{a_2-v} \cup Q_{x_1-w}$, then, for the same reason why $v$ and $w$ had to be different from $x_2$, $x$ is the center of a 4-wheel centered on $x$ in $H$ such that the rim uses exactly on marker edge.
\item If $u \in P[a_1,x_1]$, then $Q_{x_2-u} \cup P[u,v] \cup Q_{a_2-v} \cup a_2a_3 \cup P[a_3,w] \cup P[x_2,w]$ is the rim of a 4-wheel centered on $x$ in $H$ that uses exactly on marker edge.
\item If $u \in P[x_1,v[$, then $Q_{x_2-u} \cup P[u,x_1] \cup Q_{x_1-w} \cup P[w,a_3] \cup a_3a_2 \cup Q_{a_2-v} \cup P[x_2,v]$ is the rim of a 4-wheel centered on $x$ in $H$ that uses exactly on marker edge. 
\item If $u \in P]w,x_3]$, then $Q_{x_2-u} \cup P[u,a_3] \cup a_3a_2 \cup Q_{a_2-v} \cup P[v,x_1] \cup Q_{x_1-w} \cup P[x_2,w]$ is the rim of a 4-wheel centered on $x$ in $H$ that uses exactly on marker edge.
\item If $u \in P[x_3,a_3]$, then $Q_{x_2-u} \cup P[u,w] \cup Q_{x_1-w} \cup P[x_1,a_1] \cup a_1a_2 \cup Q_{a_2-v} \cup P[x_2,v]$ is the rim of a 4-wheel centered on $x$ in $H$ that uses exactly on marker edge.
\end{itemize}
This completes the proof.
\end{proof}

As a trivial corollary of Theorem \ref{3connexe} we have the following result on 4-wheel-free graphs of connectivity 3.

\begin{corollary} \label{coro3connexe}
If $G$ is a 4-wheel-free graph and $\kappa(G)=3$, then $G$ contains at least two vertices of degree 3.
\end{corollary}

\subsection{4-wheel-free graphs of connectivity 2} \label{subsection2}

\begin{theorem} \label{2connexe}
If $G$ is a 4-wheel-free graph of connectivity 2 and $F$ is an end of $G$, then either there is a vertex $v \in F$ of degree at most 3 or the end block containing $F$ is $K_{4,4}$.
\end{theorem}

\begin{proof}
Let $G$ be a 4-wheel-free graph of connectivity 2, let $F$ be an end of $G$ with $N(F)=\{a,b\}$ and let $H$ be the end block of $G$ containing $F$. If $F$ is trivial then one of the outcome of the theorem holds. So we may assume that $|F| \ge 2$ and therefore $H$ is 3-connected.

\begin{claim} \label{claim}
If $(x,C)$ is a 4-wheel of $H$, then $x \in \{a,b\}$.
\end{claim}

\begin{proof}
Let $(x,C)$ be a 4-wheel of $H$ and assume $x \notin \{a,b\}$. $C$ has to contain the edge $ab$ because it is not a 4-wheel in $G$, but since we may replace $ab$ by a path in $\overline F \cup \{a,b\}$ linking $a$ to $b$, we have a contradiction
\end{proof}

Assume first that $H$ is 4-connected. By (\ref{claim}) $H$ is an almost 4-wheel-free graph and thus, by Theorem \ref{4connexe} it is $K_{4,4}$ and one of the outcome of the theorem holds.

So we may assume that $\kappa(H)=3$. Let $E$ be an end of $H$. If $E \cap \{a,b\} = \emptyset$ then, by (\ref{claim}) and Theorem \ref{3connexe}, $E$ is trivial and we are done. So we may assume that say $a \in E$ which implies that $b \in E \cup N(E)$ and thus there is an end of $H$ included in $\overline E$ that does not intersect $\{a,b\}$ and which is trivial by Theorem \ref{3connexe}.
\end{proof}

\subsection{4-wheel-free graphs of any connectivity} \label{subsectionMain}
 Taking advantage of results obtained in  previous subsections we now prove the main result announced in the introduction.

\begin{theorem} \label{coromain}
If $G$ is a 4-wheel-free graph then either it contains a vertex of degree at most 3, or it contains a pair of twins.
\end{theorem}

\begin{proof}
If $\kappa(G) \ge 4$ (resp.\ $\kappa(G)=3$, resp.\ $\kappa(G)=2$), then the result follows by Theorem \ref{4connexe} (resp.\ Corollary \ref{coro3connexe}, resp.\ Theorem \ref{2connexe}). So we may assume that $\kappa(G)=1$. Let $H$ be an end block of $G$. Since $H$ is clearly $4$-wheel-free and $\kappa(H) \ge 2$ the result follows by applying Theorem \ref{4connexe} or Corollary \ref{coro3connexe} or Theorem \ref{2connexe} on $H$ among $\kappa(G)=4$ or $3$ or $2$.
\end{proof}

Note that interestingly we don't need induction to prove Theorem \ref{coromain}, we just go to look for the demanded structure in end blocks of the graph.

\section{Acknowledgment}
The author thanks Nicolas Trotignon for helpful and constructive discussions and for his careful reading of the manuscript.

\end{document}